\documentclass[10pt]{article}

\usepackage{fullpage,amsmath,amsthm,amssymb,bbm,tensor}
\usepackage{graphicx,color,mathtools,epstopdf}

\DeclarePairedDelimiter{\floor}{\lfloor}{\rfloor}

\DeclareMathOperator*{\gs}{gs}

\input diagxy
\xyoption{tips}
\SelectTips{cm}{}
\newtheorem{theorem}{Theorem}[section]

\newtheorem{lemma}[theorem]{Lemma}

\theoremstyle{remark}
\newtheorem{remark}[theorem]{Remark}

\theoremstyle{definition}
\newtheorem{example}[theorem]{Example}
\newtheorem{definition}[theorem]{Definition}

\newcommand{\twoset}{\mathbbm 2}
\newcommand{\R}{\mathbb R}
\newcommand{\too}{\longrightarrow}
\newcommand{\oot}{\longleftarrow}
\newcommand{\up}{\uparrow}
\newcommand{\dn}{\downarrow}

\newcommand{\ising}{\textsc{Ising}}
\newcommand{\Ising}{\mathbf{Ising}}
\newcommand{\Set}{\mathbf{Set}}
\newcommand{\Cospan}{\mathbf{Cospan}}
\newcommand{\C}{\mathbf{C}}

\newcommand{\knuth}[2]{\tensor*[_{#1}^{#2}]{\mathbb K}{}}

\author{Micah Blake McCurdy, Jeffrey Egger, Jordan Kyriakidis \\ \small \emph{Department of Physics and Atmospheric Science, Dalhousie University, Nova Scotia, Canada.}}

\date{27 August 2013}	%Revisions in response to referees' reports.
\title{Decomposition and Gluing for Adiabatic Quantum Optimization}

\begin{document}

\maketitle

\begin{abstract} 
\noindent Farhi and others~\cite{Farhi} have introduced the notion of solving NP problems using adiabatic quantum computers. We discuss an application of this idea to
the problem of integer factorization, together with a technique we call \emph{gluing} which can be used to build adiabatic models of interesting problems.
Although adiabatic quantum computers already exist, they are likely to be too small to directly tackle problems of interesting practical sizes for the foreseeable future. Therefore, we
discuss techniques for decomposition of large problems, which permits us to fully exploit such hardware as may be available. Numerical results suggest that even simple decomposition techniques may
yield acceptable results with subexponential overhead, independent of the performance of the underlying device.
\end{abstract}

\let\labelstyle\textstyle	%fixes barrish annoyance
\section{Introduction}

Adiabatic quantum computing (AQC) has been suggested by Farhi and others~\cite{Farhi} as a novel method of computation, building on earlier works concerning quantum annealing such as
\cite{Finnila} and \cite{Kadowaki}. The central idea is that of the adiabatic theorem\footnote{For a host of useful
recent references, see References~27-37 of Choi~\cite{ChoiAvoid}}, which implies that sufficiently slowly
varying quantum systems can be maintained in their ground states. This theorem can be used to underpin computation by smoothly varying a quantum system, beginning with a system (the ``initial Hamiltonian'') 
with an easily prepared ground state and ending with a system (the ``problem Hamiltonian'') whose ground state encodes the problem of interest. In this paper we will restrict ourselves to problem Hamiltonians
which are \emph{classical}, that is, which are diagonal with respect to the measurement basis used to obtain one's results; this restriction is sometimes referred to as adiabatic quantum optimisation (AQO).
The efficiency of both AQC and AQO depend sensitively on the evolution path chosen from initial to problem Hamiltonian; this path is never restricted to classical Hamiltonians even in the case of AQO. 
We shall not further discuss evolution paths in this paper, and we refer interested readers to~\cite{FarhiPaths}, to~\cite{NavidMasters}, and to~\cite{Egger}.

Broadly speaking, the encoding of a given algorithm in a form suitable for AQO translates the time-complexity of the algorithm into the space complexity of the problem Hamiltonian. This fastens our attention
on \emph{NP problems}, which are precisely those problems for which putative solutions can be verified in polynomial time. Since minimization is (mathematically) atemporal with respect to the original algorithm, 
we can attack NP problems by minimizing problem Hamiltonians associated to these verification algorithms. Thus, we can attack integer factoring through integer multiplication, satisfiability through basic 
logic gates, and subset-sum through weighted addition, to give three simple examples. This possibility accounts for much of the excitement surrounding AQC and AQO. Putative AQO devices have their own
temporal behaviour, which we do not discuss here; we merely highlight the crucual notion that the classical time complexity is translated into the adiabatic space complexity.

In the present paper, we do not address the performance or efficiency of any putative adiabatic device---we concern ourselves instead with two more quotidian tasks.
First, we discuss the encoding of classical problems as the ground states of classical Hamiltonians, in a comprehensive, self-contained fashion. We establish the key
lemma, which we call ``The Gluing Lemma'' which permits us to build up complex problems from simple ones. We illustrate this process by showing how to build a Hamiltonian whose ground state encodes integer
factoring. Second, we discuss \emph{decomposition} techniques which permit us to use adiabatic quantum hardware of a fixed size to solve problems of a larger size. It is not clear just how much ``overhead''
this decomposition entails, over and above the running time of a given adiabatic device (which, we reiterate, we do not discuss), but we present computational results which suggest it need not be exponential.

\section{Gluing of Ising Networks}

We rehearse some basic definitions to fix notation. First, let us write $\twoset = \{\up,\dn\}$ for the two element set whose elements will be known as ``up'' and ``down'', respectively, and let us write
$\twoset^X$ for the set of functions from a set $X$ to $\twoset$, that is, an assignment of up or down to every element of $X$. Given a function $s \colon X \too \twoset$ and a subset $I \too X$, we write
$\left.s\right|_I$ for the obvious restriction of $s$ to a function $I \too \twoset$; similarly, if $s_I \colon I \too \twoset$ and $s_J \colon J \too \twoset$ are two functions then we write
$\left<s_I,s_J\right>$ for the obvious function from the disjoint union $I \cup J$ to $\twoset$.

%\subsection{Ising Nets}
\begin{definition}[Ising nets]
	An \emph{Ising net} $\mathbb I$ over $\R$ consists of a set of vertices $I$ and an \emph{energy function} $E_I \colon \twoset^I \too \R$ which associates to every configuration of the network its energy. 
A \emph{ground state} of an Ising Net is a configuration $s \colon I \too \twoset$ for which $E_I(s)$ is minimal; note that ground states may or may not be unique. We write $E_{\mathbb I}^0$ for
the energy value of the ground state.
\end{definition}

We shall be chiefly interested in Ising nets which are \emph{2-local}, that is:
\begin{definition}[2-locality]
	Let $\mathbb I = (I,E_{\mathbb I} \colon \twoset^I \too \R)$ be an Ising net and let $s \colon I \too \twoset$ be a configuration of $\mathbb I$.
	Define a function from $\twoset$ to $\R$ by ${\dn} \mapsto -1$ and ${\up} \mapsto +1$, and let us write $\hat s$ for the composition of $s$ with this function.
 %Given $i,j \in I$, define functions $s_2 \colon I \times I \too \R$
%	and $s_1 \colon I \too \R$ by \[ s_2(i,j) = \left\{\begin{array}{rl} +1 & \textrm{if $s(i) = s(j)$} \\ -1 & \textrm{if $s(i) \neq s(j)$} \end{array}\right. \qquad \qquad 
%	s_1(i) = \left\{\begin{array}{rl} +1 & \textrm{if $s(i) = \up$} \\ -1 & \textrm{if $s(i) = \dn$} \end{array}\right. \]

	Then $\mathbb I$ is said to be \emph{2-local} if $E_{\mathbb I}$ can be written in the form: 
	\[ E_{\mathbb I}(s) = \sum_{i,j \in I} \beta_{i,j}\hat s(i) \hat s(j) + \sum_{i \in I} \alpha_i \hat s(i) + \gamma \] for some $\gamma$, $\alpha_i$,
	and $\beta_{i,j}$ in $\R$. Note that it is assumed that the first summation is taken over all unordered, distinct pairs of elements in $I$.
\end{definition}

We will frequently render 2-local Ising nets in a handy graphical manner. 
\begin{example}[\textsc{And} Gate] \label{and-gate}
Consider the Ising net $\mathbb A = (\{a',b',c'\},E_{\mathbb A})$ where
$E_{\mathbb A}(a'\mapsto a,b'\mapsto b,c'\mapsto c) = -a -b + 2c + ab - 2c(a+b)$, which is clearly 2-local. We depict $\mathbb A$ as:
\[ \bfig

\scalefactor{0.5}

\node A(-1000,0)[a]
\node B(+1000,0)[b]
\node C(0,-1000)[c]

\place(-1200,+200)[-1]
\place(+1200,+200)[-1]
\place(0,-1200)[2]

\arrow|a|/-/[A`B;\displaystyle 1]
\arrow|b|/-/[A`C;\displaystyle -2]
\arrow|b|/-/[B`C;\displaystyle -2]

\efig \]
The coefficients $\alpha$ appear as the labels on the vertices of this graph, and the coefficients $\beta$ appear as labels on edges. Spins $i,j$ for which $\beta_{i,j} = 0$ are not joined.

If we interpret $\up$ as ``true'' and $\dn$ as ``false'', then this net has a ground state which encodes the graph of $c = a \textsc{ and } b$. The full graph of $E_{\mathbb A}$ is
\[ \begin{tabular}[h]{ccc|c}
a & b & c & $E_{\mathbb A}$ \\ \hline
$\dn$ &$\dn$ &$\dn$ & -3 \\
$\dn$ &$\up$ &$\dn$ & -3 \\
$\up$ &$\dn$ &$\dn$ & -3 \\
$\up$ &$\up$ &$\up$ & -3 \\ \hline
$\dn$ &$\up$ &$\up$ & 1 \\
$\up$ &$\dn$ &$\up$ & 1 \\
$\up$ &$\up$ &$\dn$ & 1 \\
$\dn$ &$\dn$ &$\up$ & 9
\end{tabular} \]
\end{example}

Part of the interest in 2-local Ising nets is that, under some additional geometric restrictions, they model the class of systems for which the adiabatic quantum computer produced by D-Wave systems is suited.

We first establish the key technical lemma which we use to join together simple arithmetic operations into networks capable of computing non-trivial functions.

\begin{definition}[Gluing of Ising nets]
Let $\mathbb I = (I,E_{\mathbb I} \colon \twoset^I \too \R)$ and $\mathbb J = (J,E_{\mathbb J} \colon \twoset^J \too \R)$ be two Ising nets, and suppose that we have a pair of set inclusions $I \oot T \too J$
which describe an intersection of the sets of spins underlying the two Ising networks. Consider the set $I \underset{T}{+} J$ defined to be the union of the sets $I$ and $J$, presumed to be disjoint except for the overlap $T$. We define
the \emph{gluing of $\mathbb I$ and $\mathbb J$ along $T$} to be $\mathbb I \underset{T}{+} \mathbb J = (I \underset{T}{+} J,E_{\mathbb I \underset{T}{+} \mathbb J} \colon \twoset^X \too \R)$ by setting 
$E_{\mathbb I \underset{T}{+} \mathbb J}(s) = E_{\mathbb I}(\left.s\right|_I) + E_{\mathbb J}(\left.s\right|_J)$.
\end{definition}

The gluing of two 2-local Ising nets is again 2-local, in a very simple way; we identify the indicated spins, obtain new coefficients $\alpha$ by adding the relevant $\alpha$ coefficients, and obtain new $\beta$ values by
adding the relevant $\beta$ values. For example, the gluing of:
\[ \bfig

\scalefactor{0.5}

\node A(-1000,0)[a]
\node B(+1000,0)[b]
\node C(0,-1000)[c]

\place(-1200,+200)[-1]
\place(+1200,+200)[-1]
\place(0,-1200)[2]

\arrow|a|/-/[A`B;1]
\arrow|b|/-/[A`C;-2]
\arrow|b|/-/[B`C;-2]

\place(+2000,-700)[\textrm{and}]

%repeat, shifted.
\node A(+3000,0)[b]
\node B(+5000,0)[d]
\node C(+4000,-1000)[c]

\place(+2800,+200)[-1]
\place(+5200,+200)[-1]
\place(+4000,-1200)[2]

\arrow|a|/-/[A`B;1]
\arrow|b|/-/[A`C;-2]
\arrow|b|/-/[B`C;-2]

\efig \]
along the set $\{b,c\}$ is 
\[ \bfig

\node A(-1000,0)[a]
\node B(0,0)[b]
\node C(0,-1000)[c]
\node D(+1000,0)[d]

\place(-1200,+200)[-1]
\place(0,+200)[-2]
\place(0,-1200)[4]
\place(+1200,+200)[-1]

\arrow|a|/-/[A`B;1]
\arrow|a|/-/[D`B;1]
\arrow|b|/-/[A`C;-2]
\arrow|m|/-/[B`C;-4]
\arrow|b|/-/[D`C;-2]

\efig \]

\begin{lemma}[The Gluing Lemma for Ising Networks] \label{lem-gluing}
Let $\mathbb X = \mathbb I \underset{T}{+} \mathbb J$ be the gluing of $\mathbb I$ and $\mathbb J$ along $I \too T \oot J$ as in the previous definition.
If there exists a ground state configuration $\phi_0 \colon I \too \twoset$ of $I$ and a ground state configuration $\psi_0 \colon J \too \twoset$ of $J$ which agree on the intersection $T$, then the ground state of
$\mathbb X$ is precisely those configurations whose restrictions to $\mathbb I$ and $\mathbb J$ are each ground states of those nets.
\end{lemma}
\begin{proof}
First, consider a state $s$ of $\mathbb X$ for which $\left.s\right|_I$ is a ground state of $\mathbb I$ and for which $\left.s\right|_J$ is a ground state of $\mathbb J$---at least one such state exists by
hypothesis. To see that $s$ is a ground state of $\mathbb X$, consider another state $t$; we compute: 
\[ E_{\mathbb X}(t) = E_{\mathbb I}(\left.t\right|_I) + E_{\mathbb J}(\left.t\right|_J) \geq E_{\mathbb I}(\left.s\right|_I) + E_{\mathbb J}(\left.s\right|_J) = E_{\mathbb X}(s) \] Note that this implies that 
$E_{\mathbb X}^0 = E_{\mathbb I}^0 + E_{\mathbb J}^0$.

Conversely, suppose that $u$ is a ground state of $\mathbb X$; we must show that $\left.u\right|_I$ is a ground state of $\mathbb I$ and that $\left.u\right|_J$ is a ground state of $\mathbb J$. Suppose, for a contradiction,
that one of these is false, without loss of generality, let us suppose that $E_{\mathbb I}(\left.u\right|_I)> E_{\mathbb I}^0$, then we compute:
\[ E_{\mathbb X}(u) = E_{\mathbb I}(\left.u\right|_I) + E_{\mathbb J}(\left.u\right|_J) > E_{\mathbb I}^0 + E_{\mathbb J}(\left.u\right|_J) \geq E_{\mathbb I}^0 + E_{\mathbb J}^0 = E_{\mathbb X}^0 \] contradicting the 
assumption that $u$ is a ground state of $\mathbb X$.
%Since the ground state energy of $\mathbb X$ is the sum of the ground state energies of $\mathbb I$ and $\mathbb J$, the only way to 
\end{proof}

The assumption that $\mathbb I$ and $\mathbb J$ should have a mutually compatible ground state is, of course, necessary in the above theorem. To apply the Gluing Lemma in practice, one must build one's nets carefully,
as in our example of factoring nets in the sequel, ensuring that this condition is satisfied at all times. Of course, this is not always possible; however, the Gluing Lemma is still helpful in this case. An easy
corollary of the lemma is that $E_{\mathbb I \underset{T}{+} \mathbb J}^0 \geq E_{\mathbb I}^0 + E_{\mathbb J}^0$, with equality if and only if $\mathbb I$ and $\mathbb J$ are compatible. Hence, if an (ideal)
adiabatic evolution of a glued system gives a higher energy than expected, one can deduce that the nets in question are not compatible. The meaning of this incompatibility will vary according to circumstance,
for instance, it could mean that a given satisfiability statement is unsatisfiable, that a given number is not factorizable with factors of the desired size, or that a given set of numbers does not have a zero-sum
subset, for instance.

\subsection{Clamping}

\begin{definition}
Let $\mathbb I = (I,E_{\mathbb I} \colon \twoset^I \too \R)$ be an Ising net, and let $s \colon S \too \twoset$ be a configuration associated to a subset $S \too I$. Then the \emph{clamping of $\mathbb I$ along
$s$} is another Ising net, which we write $c_s(\mathbb I) = (I\backslash S,E_{c_s(\mathbb I)} \colon \twoset^{I\backslash S} \too \R)$. The energy function for the clamping of $\mathbb I$ along $s$ is defined by:
\[ E_{c_s(\mathbb I)}(t) = E_{\mathbb I}(\left<s,t\right>) \] for a configuration $t \colon I\backslash S \too \R$.%perhaps change the angle brackets to +?
\end{definition}

It is straightforward but pleasing to verify that the clamping of a 2-local Ising net is again 2-local.
Ising nets can be used to model computations in the following manner. 
\begin{definition}[Programs on Ising nets] \label{def-program}
Let $\mathbb I = (I,E_{\mathbb I} \colon \twoset^I \too \R)$ be an Ising net. A \emph{program} on $\mathbb I$ is an ordered pair of subsets, $(S,T)$ of $I$. 
We think of the set $S$ as the type of ``input'' of the program and the set $T$ as the type of the ``output'' of the program.
Note that, for technical convenience, we do not assume that $S$ and $T$ are disjoint, although in most programs this will be the case.
\end{definition}

One virtue of this approach to computation is that is \emph{atemporal}, that is, the choices of ``source'' vertices $S$ and ``target'' vertices $T$ is completely arbitrary. For example, consider the \textsc{and} gate
from Example~\ref{and-gate}. Setting $S = \{a,b\}$ and $T = \{c\}$, we can compute the logical \textsc{and} of $a$ and $b$. Conversely, setting $S = \{c\}$ and $T =\{a,b\}$, we can compute the set of pairs $a,b$
for which $a \textsc{ and } b$ equals the given value of $c$.

\begin{definition}[Executions of programs on Ising nets]
Let $(S,T)$ be a program on an Ising net $\mathbb I$. 
An \emph{execution of the program $(S,T)$} is the following procedure: Given a configuration $s \colon S \too \twoset$ (that is, input); obtain the clamping $c_s(\mathbb I)$. Minimizing $E_{c_s(\mathbb I)}$
produces a set of configurations of the form $t \colon I\backslash S \too \twoset$, combining these with the given $s$ produces a set of configurations of $I$, from which configurations of $T$ may be extracted, such
configurations are the output of the execution.

\end{definition}
\begin{example}[The Full Multiplier] \label{Knuth-unit} Consider the net displayed in Figure~\ref{fig-full-multiplier}. Its ground state encodes (as the reader may verify) the graph of the relation $ab + c + d = 2e + f$.
The program $(\{a,b,c,d\},\{e,f\})$ defined on this Ising net takes as input a quadruple of binary values $a,b,c,d$ and computes from them the expression $ab+c+d$, rendered as the singleton set $\{e,f\}$ of the digits
of this number in binary. On the other hand, the program $(\{c,d,e,f\},\{a,b\})$ defined on this net will take a quadruple of binary values $c,d,e,f$ and produce the set of all pairs $(a,b)$ whose product $ab = 2e+f-c-d$.
Of especial interest is the program $(\{e,f\},\{a,b,c,d\})$ which is the ``time reversal'' of the first program, where we infer quadruples $a,b,c,d$ for which $ab+c+d$ is the same as the number $(ef)_2$. The reader may
verify that the ground state energy of this net is $-15$.
\end{example}

\begin{figure}[h!]
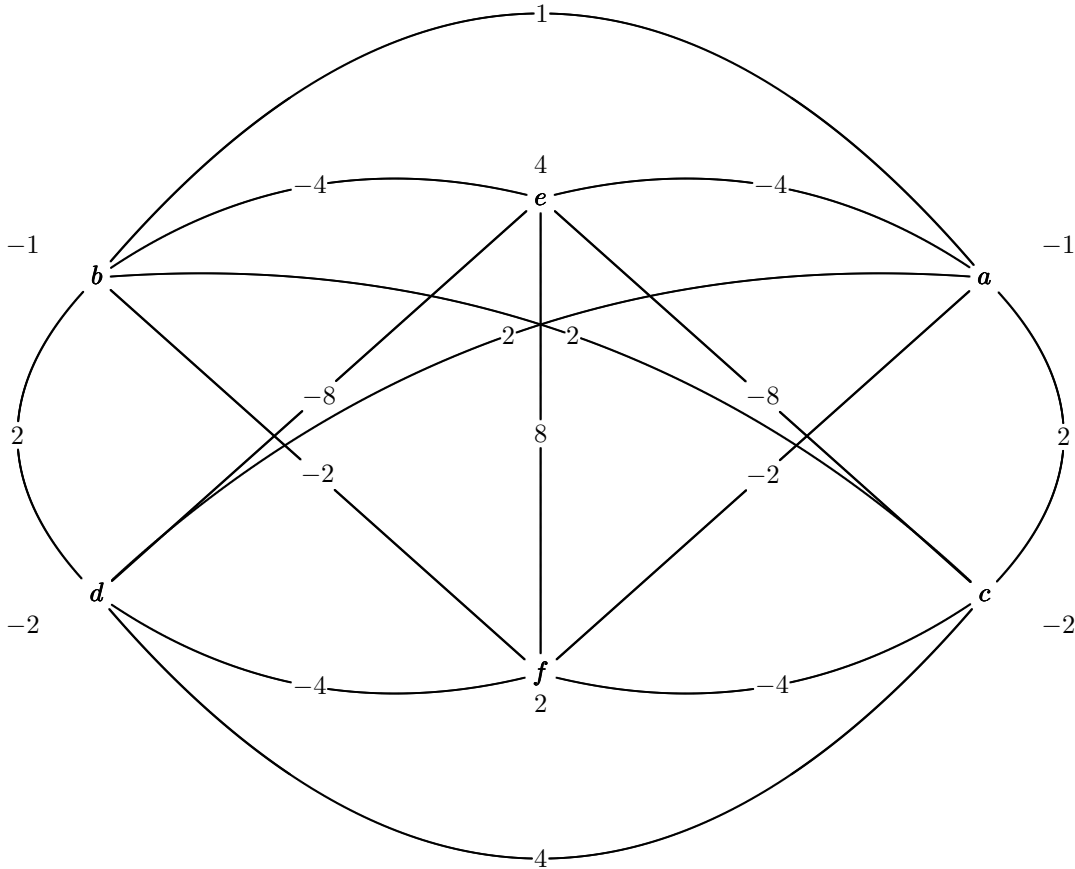

\[ \bfig

\hscalefactor{1.4}
\vscalefactor{0.6}

\node 3(0,+1500)[e]
\node 0(0,-1500)[f]
\node 2(-1200,+1000)[b]
\node 1(+1200,+1000)[a]
\node 5(+1200,-1000)[c]
\node 4(-1200,-1000)[d]

\place(0,+1700)[4]
\place(0,-1700)[2]
\place(-1400,+1200)[-1]
\place(+1400,+1200)[-1]
\place(+1400,-1200)[-2]
\place(-1400,-1200)[-2]

\arrow|m|/{@{-}@/_0em/}/[0`1;-2]
\arrow|m|/{@{-}@/^0em/}/[0`2;-2]
\arrow|m|/{@{-}@/_2em/}/[0`5;-4]
\arrow|m|/{@{-}@/^2em/}/[0`4;-4]

\arrow|m|/{@{-}@/_0em/}/[3`0;8]

\arrow|m|/{@{-}@/_2em/}/[3`2;-4]
\arrow|m|/{@{-}@/^2em/}/[3`1;-4]
\arrow|m|/{@{-}@/_0em/}/[3`4;-8]
\arrow|m|/{@{-}@/^0em/}/[3`5;-8]

\arrow|m|/{@{-}@/_10em/}/[1`2;1]
\arrow|m|/{@{-}@/_10em/}/[4`5;4]

\arrow|m|/{@{-}@/^3em/}/[1`5;2]
\arrow|m|/{@{-}@/^3em/}/[4`2;2]

\arrow|m|/{@{-}@/^4em/}/[2`5;2]
\arrow|m|/{@{-}@/_4em/}/[1`4;2]

\efig \]
\caption{The Full Multiplier. This net has a ground state which encodes the graph of the equation $ab + c + d = 2e + f$}
\label{fig-full-multiplier}
\end{figure}

%\begin{remark}[Ground state energy of the 

If we are to solve serious problems in this way we must have a method for building up non-trivial programs from simple ones, we do this with the Gluing~Lemma.
\begin{definition}[Composition of programs] \label{def-composition}
Let $(S,T)$ be a program on $\mathbb I$ and let $(T,U)$ be a program on $\mathbb J$. We say that these two programs are \emph{compatible} if there is a ground state of $\mathbb I$ and a ground state of $\mathbb J$ whose
restrictions to $T$ are equal. In this case, we define the \emph{composition} of these two programs to be $(S,U)$ on the gluing $\mathbb I \underset{T}{+} \mathbb J$ along $T$.
\end{definition}

The categorically-minded reader will no doubt have detected that this composition is reminiscent of the composition in a suitable bicategory of cospans of sets. We will not have much to do with this categorical structure, 
but the interested reader may consult the Appendix.

Since the composition of two programs is only defined when they are compatible, we see from the Gluing~Lemma that executions of composite program behave as we expect, that is, producing output in the form of states which
simultaneously satisfy both programs.

\section{Factoring}

With the general framework of the previous section in hand, we apply these concepts to building a net whose ground state encodes integer factoring.

The Ising net $\mathbb K$ in Example~\ref{Knuth-unit} encodes the graph of the function $ab+c+d=2e+f$. Knuth~(\cite{Knuth}, p268) shows how to use this function to construct a multiplication algorithm, which he calls
``Algorithm M''. We
rehearse this algorithm, converted into a program on an Ising net. Let us write $r = (r_nr_{n-1}\ldots r_3r_2r_1)_2$ for the binary representation of an $n$-bit number and $g = (g_mg_{m-1}\ldots g_3g_2g_1)_2$ for that
of an $m$-bit number, we describe the multiplication of $r$ and $g$ to produce an $m+n$-bit product.

\begin{definition}[Knuth nets]

Recall from Example~\ref{Knuth-unit} the net on six vertices $\{a,b,c,d,e,f\}$ whose ground state encodes $ab+c+d=2e+f$; we depict it schematically as:
\[ \includegraphics[scale=2]{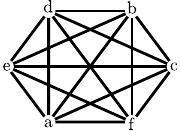} \] 
suppressing the field and coupling terms. In fact, since we will have no need of any other orientations of the six vertices, we shall also suppress these letters.

For each $(i,j)$ satisfying $1 \leq i \leq n$ and $1 \leq j \leq m$, we consider a net $\mathbb K_{i,j} = (\{a_{i,j},b_{i,j},c_{i,j},d_{i,j},e_{i,j},f_{i,j}\},E_{\mathbb K})$, where $E_{\mathbb K}$ is as
in Example~\ref{Knuth-unit}. For example, if we take $m = 4$ and $n = 3$, we have 12 copies of this net. For our convenience, we shall arrange them in a grid with the origin at the top-right: 
\[ \includegraphics{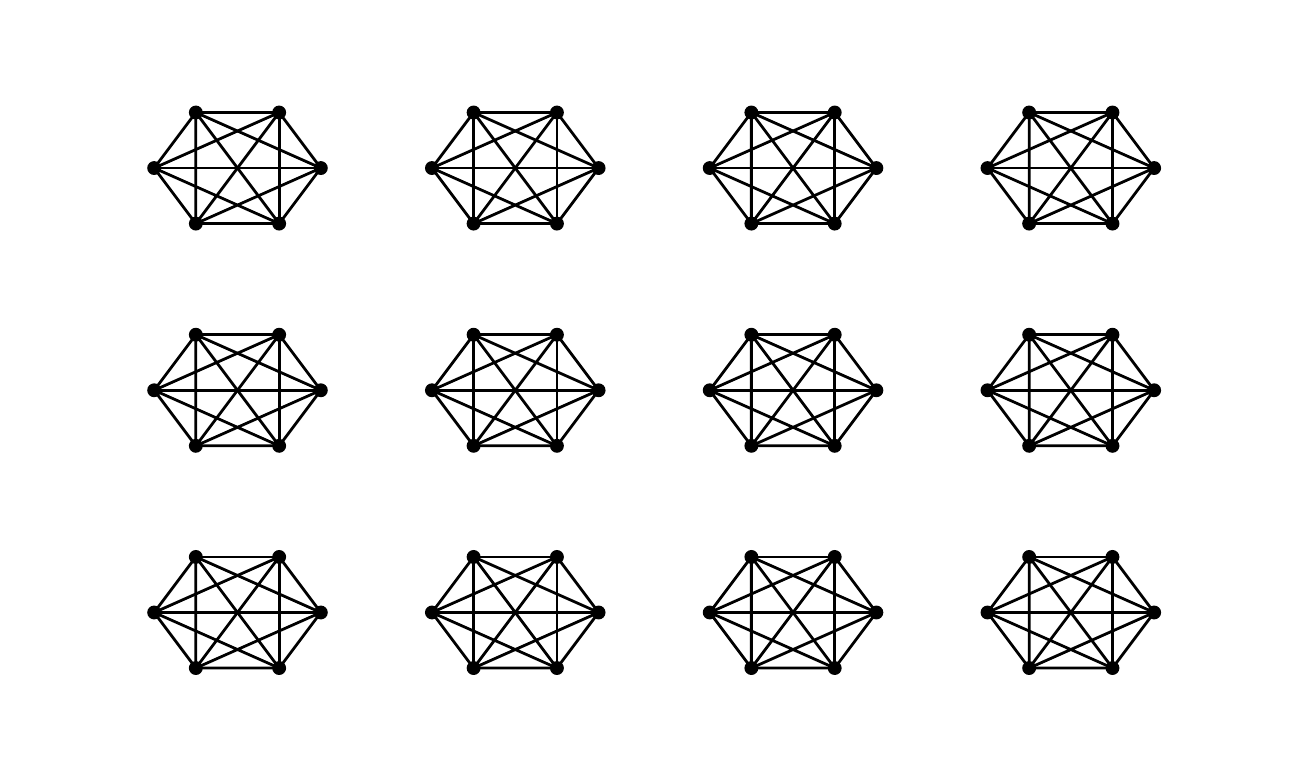} \]

Among these $ij$ unlinked Ising nets, we will perform
various identifications to make a single net; we indicate these identifications with coloured edges. These edges are \emph{not} couplings, they indicate that the two vertices so linked are to be thought of as one vertex.
To link up with the formalism of the preceding, a coloured edge from a vertex $a$ in an Ising net $\mathbb A = (A,E_{\mathbb A})$ to an vertex $b$ in an Ising net $\mathbb B = (B,E_{\mathbb B})$ is the gluing of
$\mathbb A$ with $\mathbb B$ along the pair of functions $\{\ast\} \too A$ and $\{\ast\} \too B$ defined by $\ast \mapsto a$ and $\ast \mapsto b$ respetively.

First, we identify all vertices of the form $a_{i,j}$ with the symbol $r_i$: \[ \includegraphics{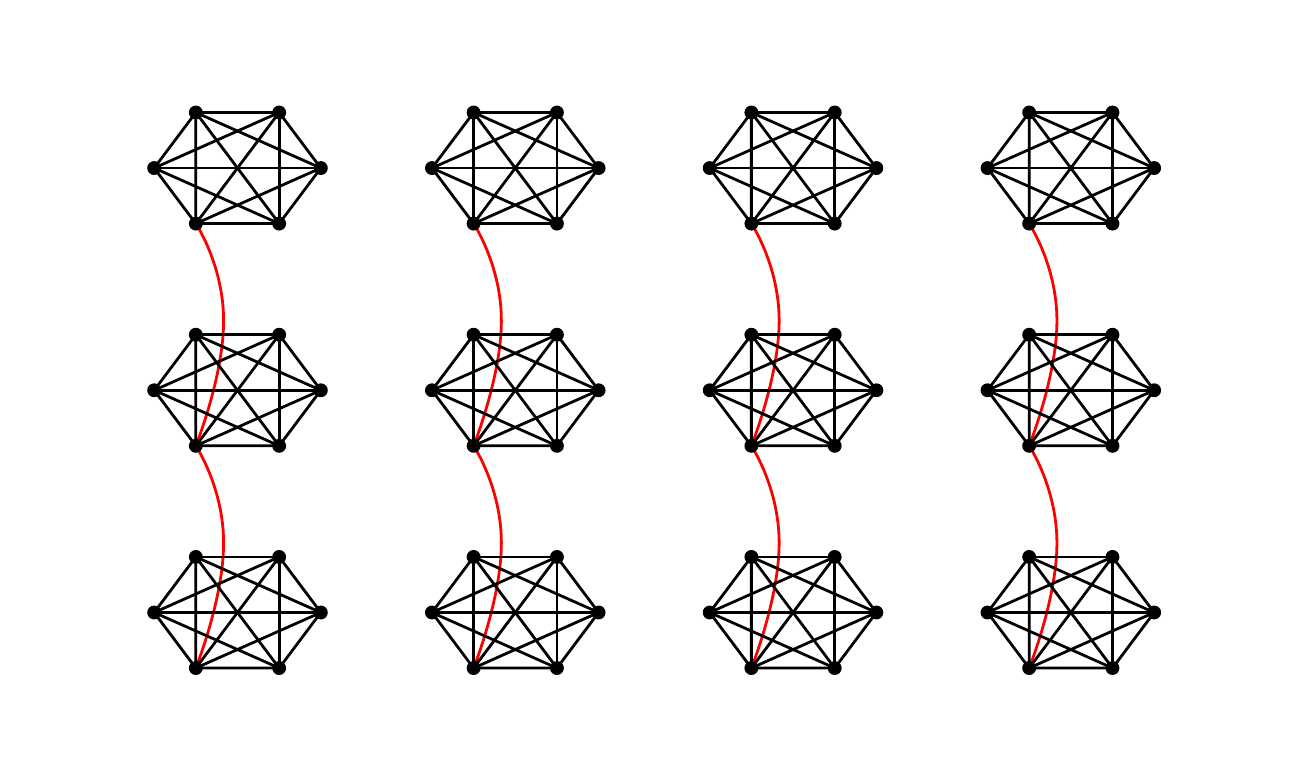} \]
and we identify all vertices of the form $b_{i,j}$ with the symbol $g_j$: \[ \includegraphics{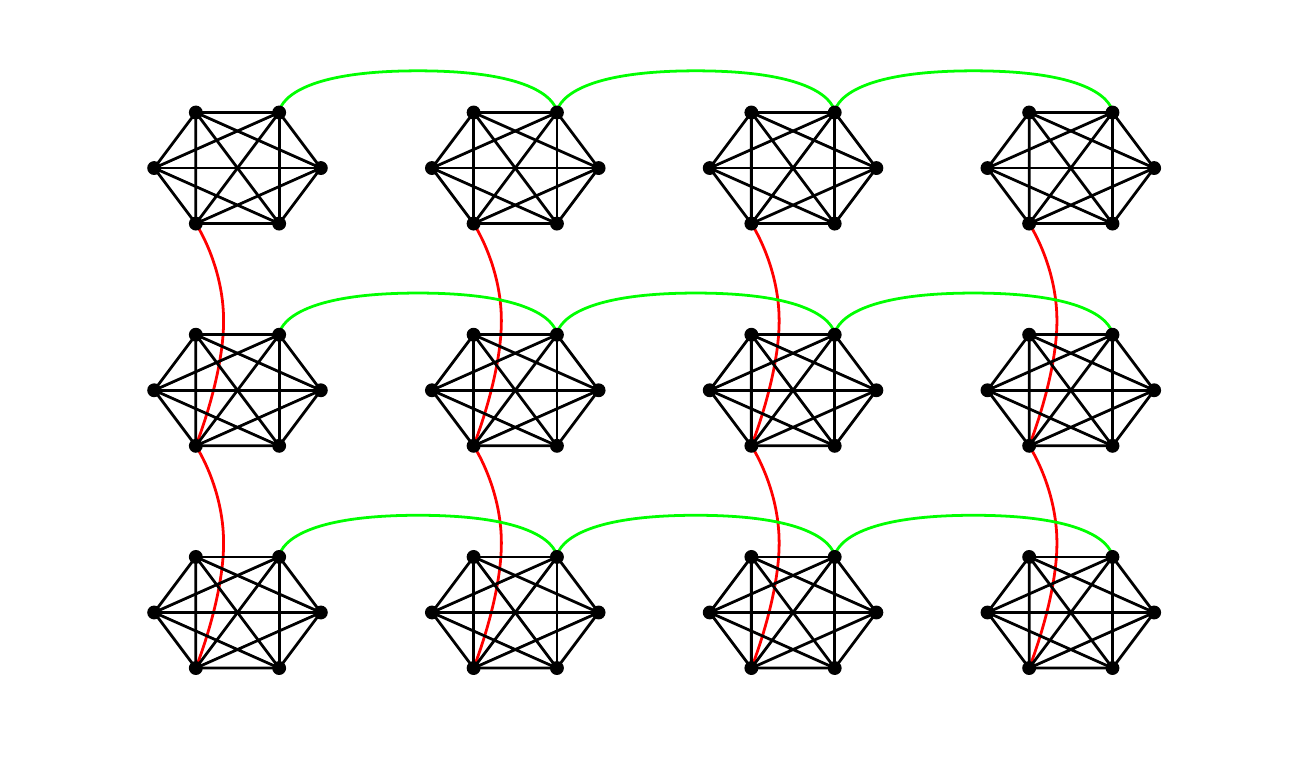} \]

Next, for each $j$ and for
each $i < n$, we identify the symbol $c_{i,j}$ with $e_{i+1,j}$: \[ \includegraphics{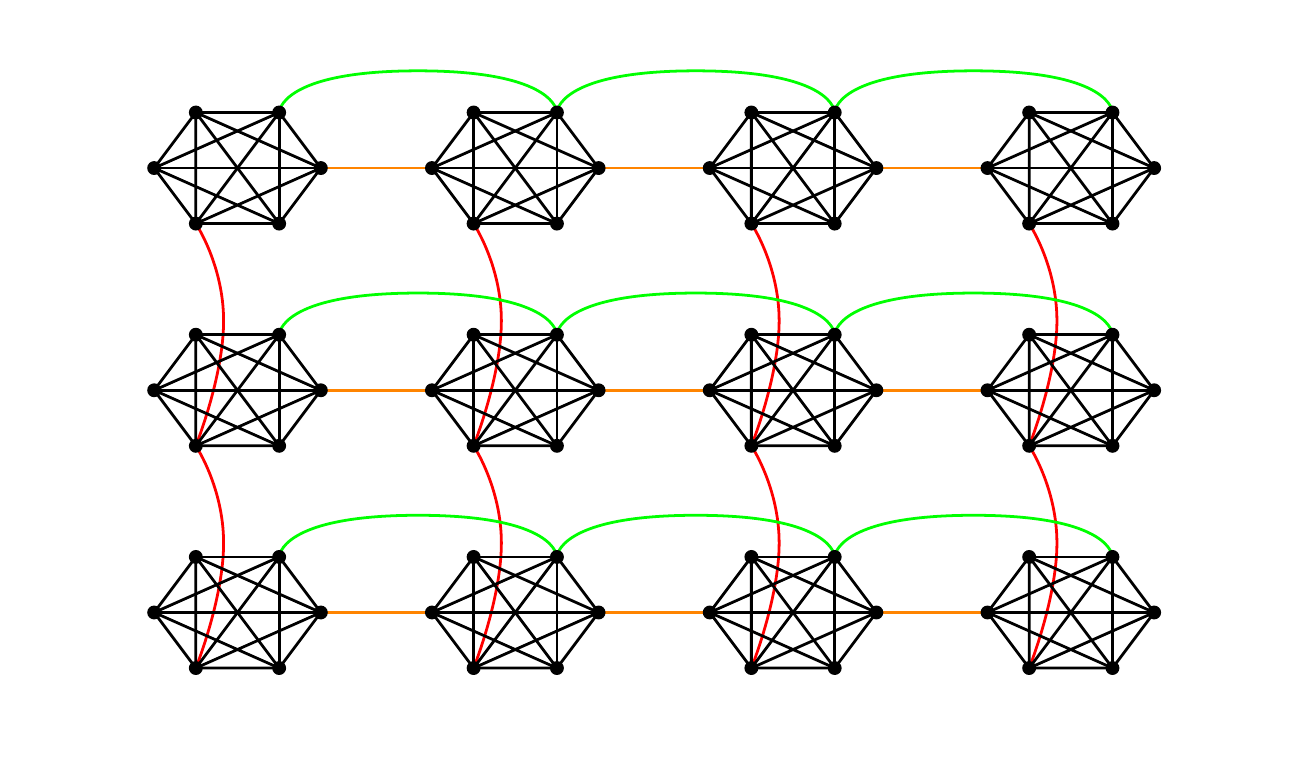} \]

Furthermore, for each $i < n$ and each $j < m$, we identify the symbol $d_{i,j+1}$ with $f_{i+1,j}$: \[ \includegraphics{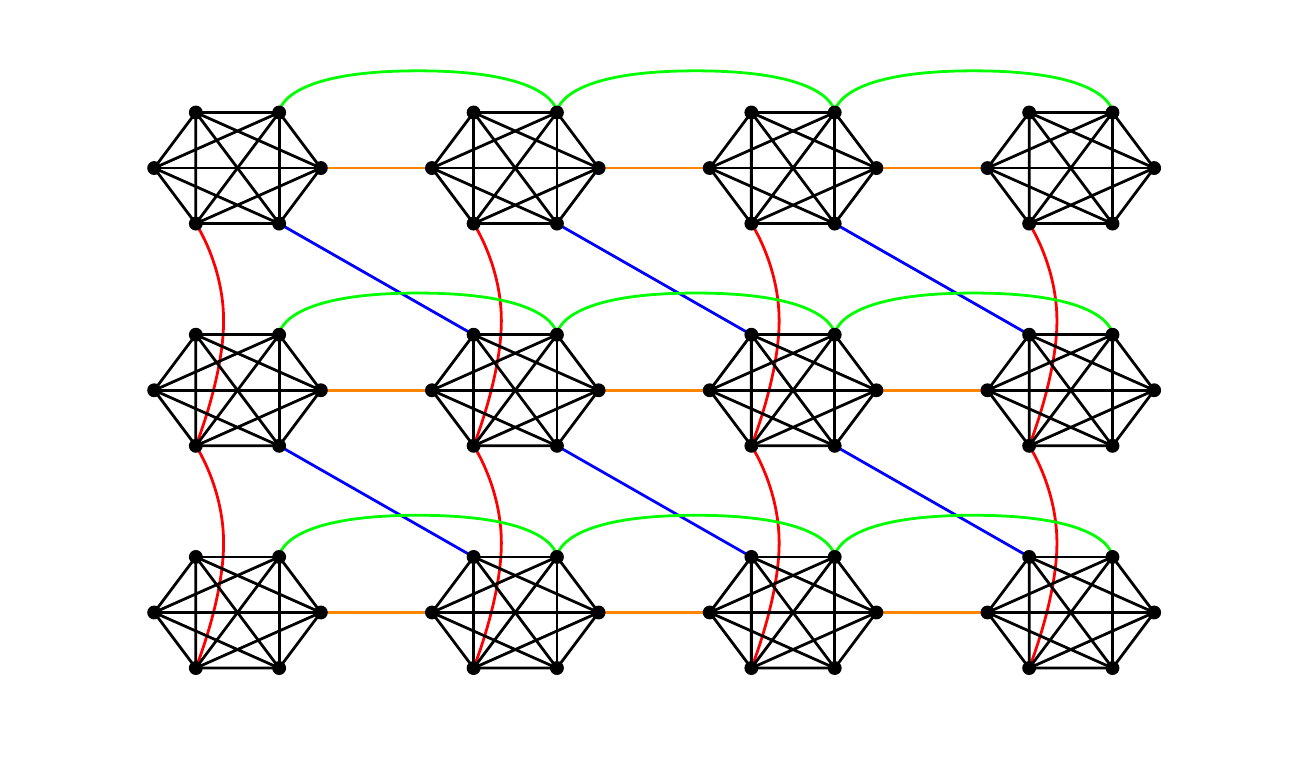} \]

These last two sets of identifications comprise Knuth's ``M4''.

To complete the identifications, we identify $e_{n,j}$ with $d_{n,j+1}$ for all $j < m$: \[ \includegraphics{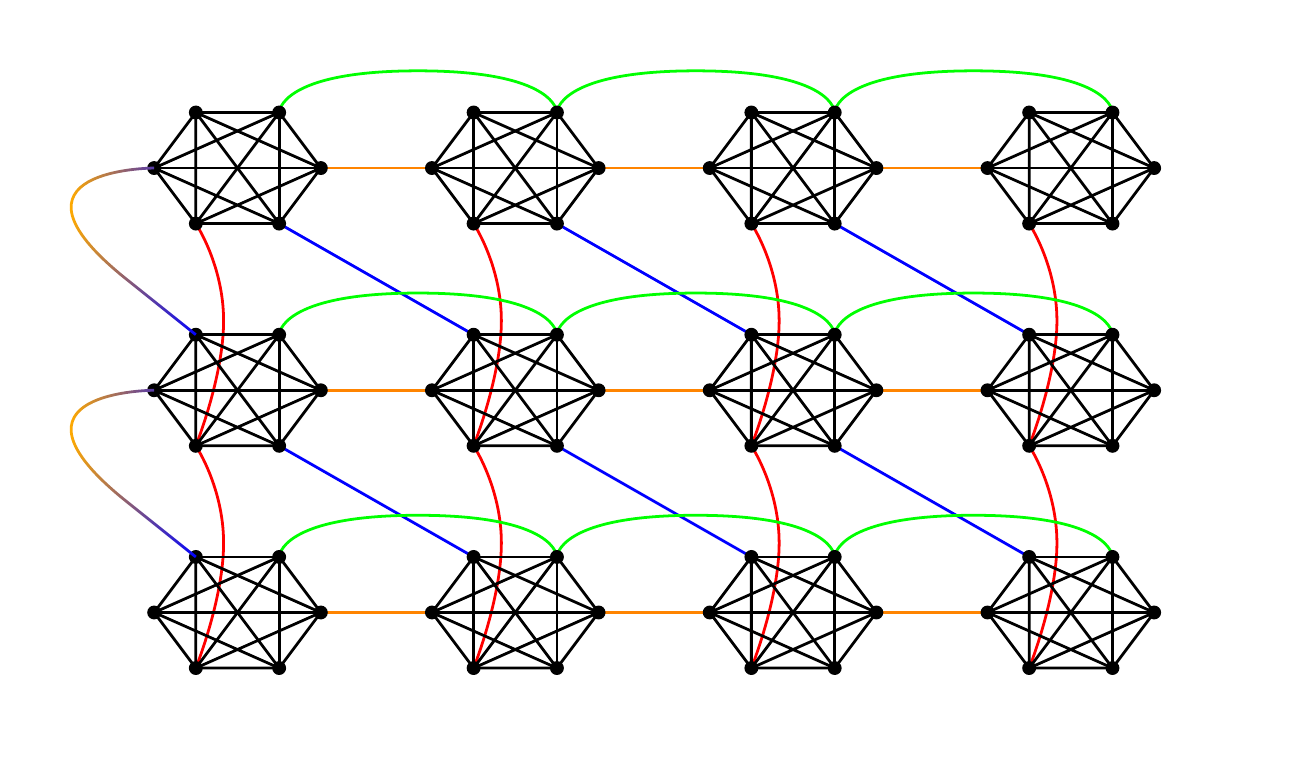} \]

This is Knuth's ``M5''. Finally, we clamp all of the symbols of the form $d_{i,1}$ or $c_{1,j}$ to ``$\dn$'': \[ \includegraphics{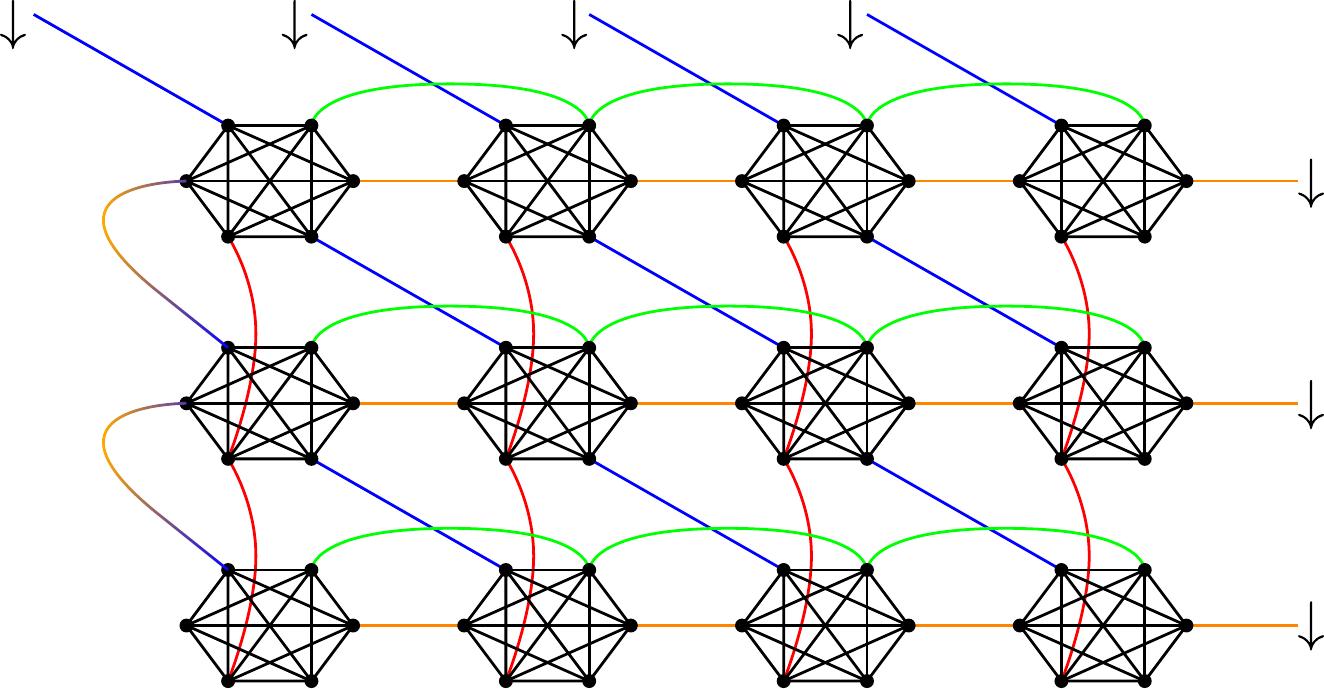} \]

This is Knuth's ``M1'', corresponding to initializing the relevant registers to zero.% \[ \includegraphics{multiplication-red-green-orange-blue-special-zero.pdf} \]

In this way, we obtain a net which we call a \emph{Knuth multiplication net}, or, simply a Knuth net, which we write $\knuth{n}{m}$. The output of Algorithm~M is the string
\[ (e_{n,m},f_{n,m},f_{n-1,m},f_{n-2,m},\ldots,f_{2,m},f_{1,m},f_{1,m-1},f_{1,m-2},\ldots,f_{1,2},f_{1,1}) \] which is the big-endian binary representation of the product $rg$. If we highlight this
string in our example net, we have: \[ \includegraphics{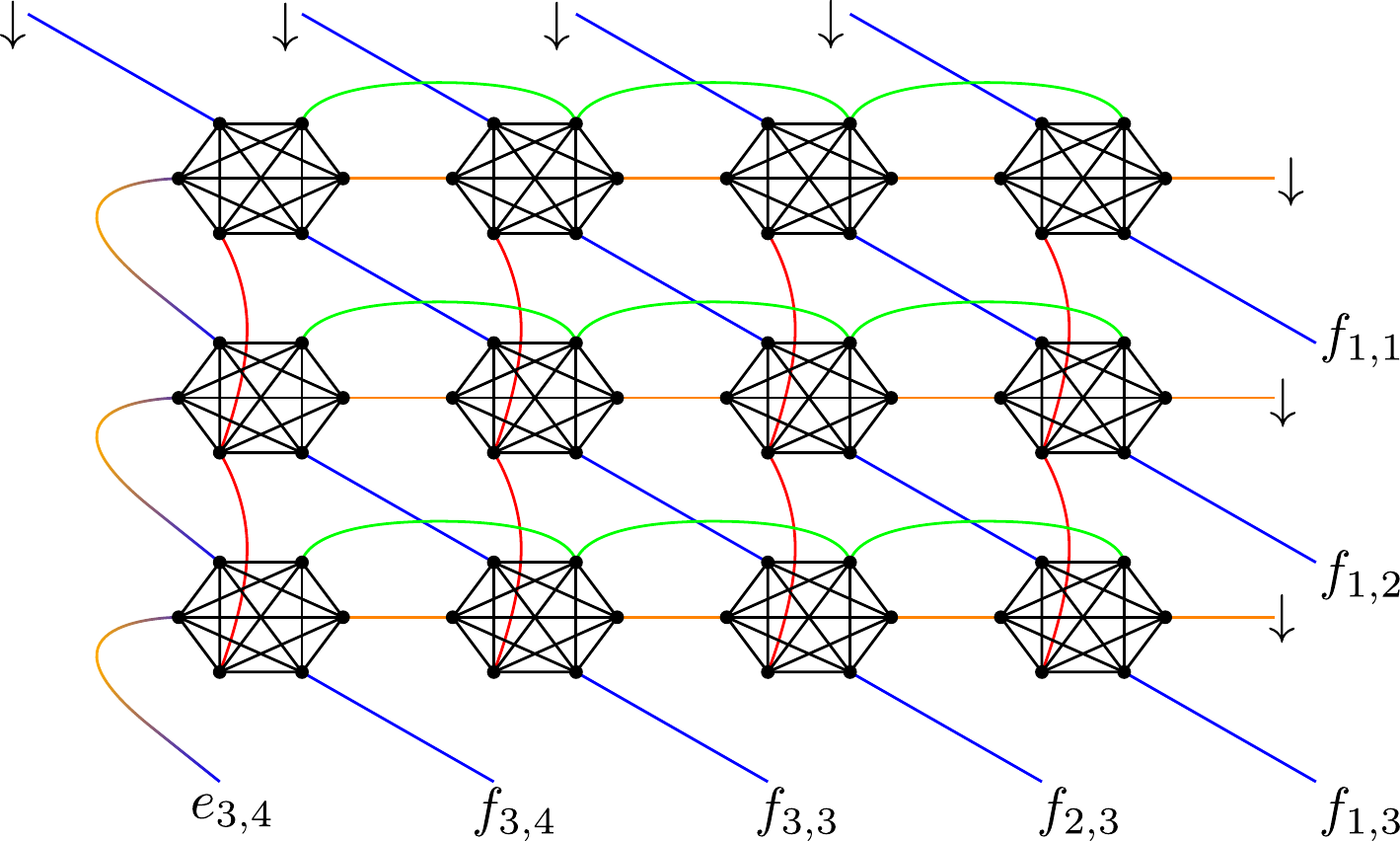} \]
\end{definition}

\begin{definition}[Factoring using Knuth nets] The purpose of Knuth nets is that we can define a program on a Knuth net, the execution of which accomplishes integer factoring.
We define a program on this
Ising net consisting of input set $(e_{n,m},f_{n,m},f_{n-1,m},f_{n-2,m},\ldots,f_{2,m},f_{1,m},f_{1,m},f_{1,m-1},f_{1,m-2},\ldots,f_{1,2},f_{1,1})$ and output set $\{r_i | 1 \leq i \leq m\} \cup \{g_j | 1 \leq j \leq n \}$. 
This is precisely the reversal 
of the net thought of as a multiplication algorithm; hence its suitability for factoring. An execution of this program is a clamping of the input set to the binary representation of an $m+n$-bit number to be factored,
followed by a minimization of the energy function associated to this clamped net, finished with a reading of the values of $r$ and $g$.
\end{definition}

The central difference between Knuth's Algorithm~M and our Ising net version thereof is the atemporal nature of the Ising net. This has great advantages---especially the suitability to minimization---but it also has
drawbacks; for instance, Algorithm~M maintains a set of registers $w_k$ for $1 \leq k \leq m+n$ which are zero ``at first'' and which hold the desired product ``at the end'', having taken on many different values
``through the course'' of the computation. All of the quoted phrases have no meaning in the Ising version; accordingly, instead of $m+n$ internal variables (which are also output!) we must instead maintain many more
auxilliary spins, that is, all of the $c_{i,j}$, $d_{i,j}$, $e_{i,j}$, and $f_{i,j}$, a few of which are zero, a few of which are the desired product, and most of which are simply auxilliary. Broadly speaking, the
time complexity of Algorithm~M becomes the space complexity of our Ising net.

\begin{remark}[Size of Knuth nets]
	The reader may readily verify that the size of the underlying set of the Knuth net $\knuth{m}{n}$ is $2mn + m + n$, thus, any execution of a factoring program on a Knuth net (whose input set is of size $m+n$) involves 
	the minimization of an energy function defined on $2mn$ vertices.
\end{remark}

Recalling the discussion after the proof the Gluing Lemma, this Knuth net will only be suitable for factoring the number $(e_{n,m}f_{n,m}f_{n-1,m}f_{n-2,m}\cdots f_{2,m}f_{1,m}f_{1,m}f_{1,m-1}f_{1,m-2}\cdots f_{1,2}f_{1,1})_2$ 
if this number has a factorization into factors of the given size; alternatively, it can be used to decide if such a factorization exists. Strictly speaking, the input to the factoring algorithm described here is not merely 
the $n+m$-bit number to be factored, but also the sizes $n$ and $m$ of the factors to be obtained. In general, given a composite $p$-bit integer
with factors of unknown sizes, we must consider a net of dimensions $\floor{p/2}$ by $p-2$ or larger; thus, a general factoring method requires a minimizing of a binary function on $O(p^2)$ variables.

\begin{remark}[Known ground state energies]
One pleasant feature of the gluing theorem is that the ground state energy of the glued net can be obtained by adding the ground state energies of the constituent nets. Thus, we know that the ground state energy of
$\knuth{m}{n}$ is given by \[ E^0_{\knuth{m}{n}} = mnE^0_{\mathbb K} = -15mn \]
\end{remark}

\section{Decomposition techniques}

Much effort has been given to trying to make sense of what the ``adiabatic running time'' of an adiabatic algorithm should be; indeed, much effort has been spent to produce a sensible notion of
what an ``adiabatic algorithm'' is---see, for instance, \cite{ChoiAlg},~\cite{Egger},~\cite{FarhiPaths}, merely to whet the appetite. However, quite aside from such considerations, we must confront the fact that the only 
existing candidate for an adiabatic quantum computer~\cite{DWave} comprises only 512 spins.
Using factoring nets of dimension $\floor{p/2}$ by $p-2$, and setting aside geometric restrictions (for discussion of which the reader may consult, for instance,~\cite{ChoiI},~\cite{ChoiII}, or~\cite{Klymko}), 512 spins can 
factor any composite number with no more than 23 bits---that is, a number as big as 8,388,608---which is hardly cryptographically fascinating 
Though technological and scientific progress will doubtless continue apace, it seems safe to assume 
that practical problems of all kinds (not merely factoring) will be comprehensively larger than available hardware for the foreseeable future. Thus, we turn our attention to \emph{decomposition} techniques, that is, methods
by which minimization problems over large sets can be broken down into smaller ones.

\subsection{How good can decomposition be?}

Let us call the problem of determining the ground state of an Ising net by the name \ising. Suppose we fix an algorithm for \ising\ whose running time for an Ising net of $n$ spins is $\Theta(f(n))$ for some function $f$. 
Since \ising\ is known to be NP-complete\cite{Barahona}, we expect that $f$ will be exponential---this is the Exponential Time Hypothesis. 
Let us suppose that we have an oracle for \ising\ when given Ising nets with no more than $n/2$ spins, and let us denote by $d(n)$ the minimum number of times this oracle must be called in any algorithm for \ising\ when
given Ising nets of size $n$. We call the function $d$ the ``decomposability'' of \ising\ and we would like to bound it somehow. Even without invoking this oracle, we have that $f(n)$ is $\Theta(d(n)f(n/2))$, whence
$f(n)$ is in $\Theta[d(n)d(n/2)d(n/4)\cdots d(1)]$ and hence $f(n)$ is in $O[d(n)^{\log_2(n)}]$ since $d$ is clearly increasing. Hence, since we expect $f$ to be exponential, we see that $d$ is in $\Omega[\exp(n/\log_2 n)]$,
which is superpolynomial. Thus, although no general decomposition algorithm can be expected to be polynomial, we have some reason to hope that it might be subexponential, in the sense that $\log d$ is $O(p)$ for any polynomial.
We reiterate that we are not considering the complexity of any adiabatic device, in theory or practice, but merely the cost of decomposition itself.

One common approach which produces good approximate solutions (that is, configurations whose energy is very close to the ground-state energy) is the class of so-called ``local update'' or ''iterated conditional mode'' algorithms.

\begin{definition}[Local Update Algorithms]\label{localupdatealgorithm}

Let $\mathbb I = (I,E_{\mathbb I})$ be an Ising net. A \emph{local update algorithm} for $\mathbb I$ proceeds as follows:
\begin{enumerate}

\item[0] Obtain an initial configuration $x = x_0 \colon I \too \twoset$.
\item[1] Select a ``figure'', that is, a subset $S \subseteq I$. 
\item[2] Form the clamping $c_{\left.x\right|_{I\backslash S}}\left(\mathbb I\right)$ of $\mathbb I$ which clamps everything outside of the figure to its current value under~$x$. Minimize the energy function associated to 
	this clamped net, obtaining a configuration $y \colon S \too \twoset$.
\item[3] Update the assignment $x \colon I \too \twoset$ by redefining $x(s) = y(s)$ for all $s \in S$; this lowers (or possibly merely preserves) the value of $E(x)$.
\item[4] Return to Step~1.

\end{enumerate}

\end{definition}

This process is repeated as desired; under certain conditions, bounds can be given on the quality of the approximations obtained in terms of the number of iterations performed.
For instance, Jung, Kohli, and Shah~\cite{JungKohliShah} give one version of such an algorithm where an Ising net $\mathbb I = (I,E_{\mathbb I})$ with $|I| = n$ can be solved within an error of $\epsilon$ by taking 
$O(\epsilon)n^2\log n$ iterations.
However, their approach relies on certain geometric assumptions about the structure of $E_{\mathbb I}$ which do not apply to our Knuth networks; moreover, we seek global ground states, and not merely low energy states.
 We are nevertheless emboldened to seek a local update algorithm the performance of which (measured by $d(n)$ above) we hope will be subexponential. For an illustrative example, we have implemented a local update
alorithm using Knuth nets--specifically, in Step~0 of Definition~\ref{localupdatealgorithm}, we choose a random assignment $x_0$, in Step~1 we select figures randomly with size half that of the net, and then in Step~2 we 
randomly choose a ground state $y$ from the (generally degenerate) ground state of the figure. We call this algorithm ``Random Half-size Local Updates''. Our measurement of the decomposability, $d(n)$, of Knuth nets of 
size $n$ is shown in Figure~\ref{decomposabilityfig}, and is gently consistent with $d$ being subexponential. 

\begin{figure}[hbt!]
\includegraphics[width=\textwidth]{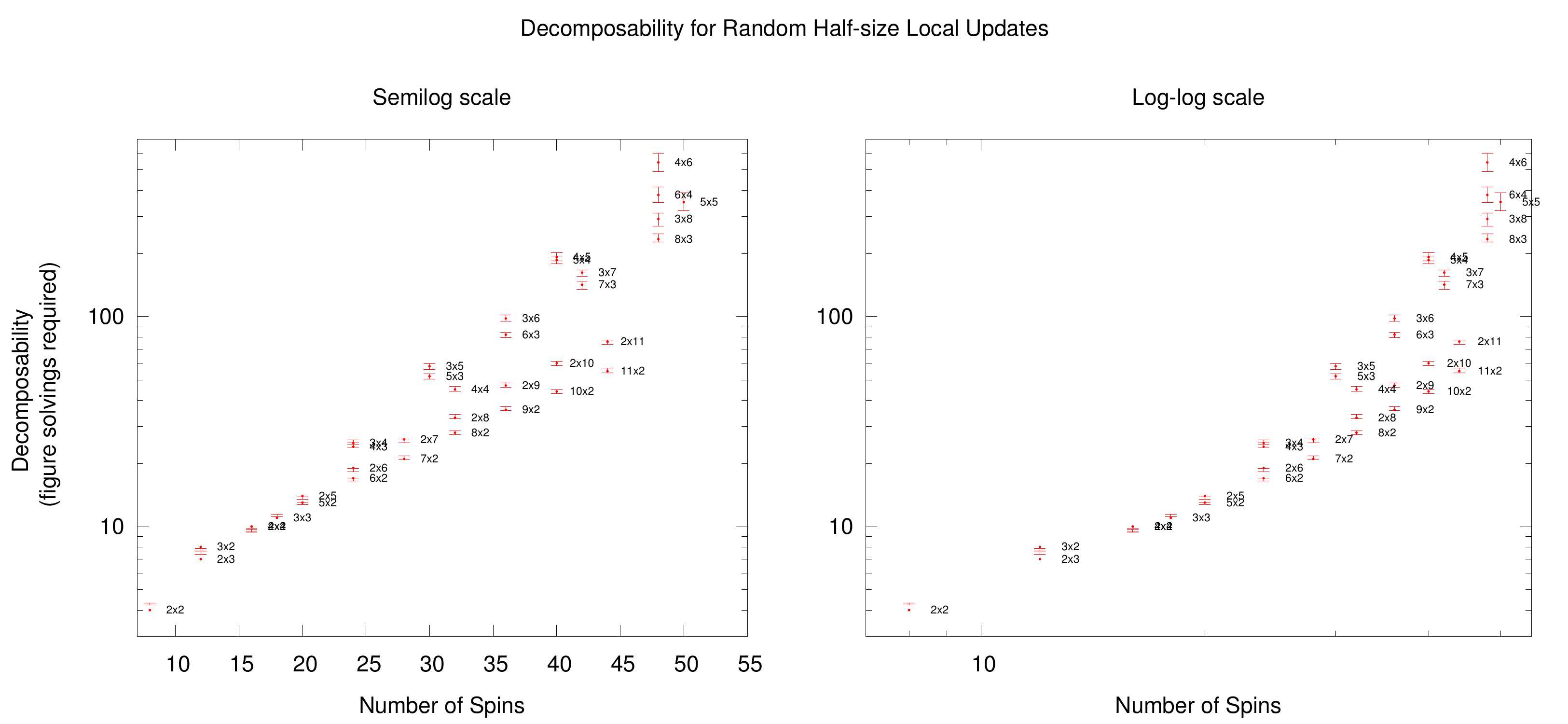}
\caption{Decomposability for Knuth Nets using Random Half-Size Local Updates. Each data point represents the median number of local updates to obtain the global ground states from 10,000 runs at each net size. A point
labelled ``$a$ x $b$'' is a Knuth net with $a$ columns and $b$ rows of full multipliers. The error bars are 95\% confidence intervals for these medians, computed using smoothed bootstraps from the sample data itself.
Note the slight concave-down trend in the semi-log scale, and the slight concave-up trend in the log-log scale, consistent with $d$ being subexponential.}
\label{decomposabilityfig}
\end{figure}

\subsection{Figure generation} The above dataset was generated in a very naive way, to focus attention on the \emph{general} problem of decomposition. However, a practitioner with a specific problem to solve will doubtless
employ more sophisticated techniques. For instance, even without leaving the realm of local update algorithms, one could choose figures using problem-specific knowledge. For instance, in our factoring
example above, each full-multiplier unit can be quickly checked to see if it is in a (local) ground state; the overall ground state is characterized as simultaneously satisfying all such full-multipliers. Spins in 
full-multiplier units which are \emph{not} satisfied are immediately suspect, since at least one of these spins must be flipped to reach the overall ground state. In future work, we intend to examine how this, and more
sophisticated number-theoretic techniques can be used to improve our decomposition techniques. An obvious practical choice is to choose figures which can easily be embedded on existing quantum hardware, for discussion
of which see~\cite{ChoiI}~and~\cite{ChoiII}.

\section{Acknowledgements}
The authors acknowledge the financial support of the Lockheed Martin Corporation.

\bibliographystyle{abbrv}
\bibliography{biblio}

\section*{Appendix}
The categorically-minded reader will have detected a categorical flavour to our definition of programs on Ising nets, and especially to our definition of composition of programs. In this appendix we briefly show that there is
a suitable category of Ising nets, cospans in which model programs on Ising nets in the sense we introduce in this paper.

As before, we define:
\begin{definition}[Ising nets]
An \emph{Ising net} $\mathbb I$ is a pair $(I,E_{\mathbb I} \colon \twoset^I \too \R)$, where $I$ is a set and a function $E_{\mathbb I}$ which associates to each configuration of $I$ its \emph{energy} in $\R$. Furthermore
to specify an Ising net we must specify a designated subset, $\gs \mathbb I \subseteq \twoset^I$, for which $E_{\mathbb I}$ is minimal, the configurations in this subset comprise the \emph{ground state} of the Ising
net $\mathbb I$.
\end{definition}

\begin{definition}
Let $\mathbb I = (I,E_{\mathbb I})$ and $\mathbb J = (J,E_{\mathbb J})$ be Ising nets. A \emph{morphism of Ising nets} $f$ from $\mathbb I$ to $\mathbb J$ is a function (abusively also called $f$) from $I$ to $J$ for
which restriction along $f$ preserves ground states.
%together with a morphism $\gs f \colon \gs \mathbb I \too \gs \mathbb J$ for which the following square:
%$\bfig \square/->` >->` >->`->/[\gs \mathbb J`\gs \mathbb I`\twoset^J`\twoset^I;\gs f`i_{\mathbb J}`i_{\mathbb I}`\twoset^f] \efig$
%is a pullback. Essentially, we demand that morphisms of Ising nets tightly preserve and reflect ground states, as much as possible.
\end{definition}

With the evident compositions and identities, we have a category of Ising nets, which we write $\Ising$.

\begin{lemma} Let us write $\Set$ for the category of sets and monomorphisms between them. The obvious forgetful functor $U \colon \Ising \too \Set$ is a fibration.
\end{lemma}

\begin{proof} Given a morphism $f \colon B \too U\mathbb A = A$ in $\Set$, simply define $\mathbb B_f = (B,E_{\mathbb B_f})$ by setting $E_{\mathbb B_f}(b) = 0$ if $b$ can be written as $a \circ f$ for $a \in \gs A$,
and $E_{\mathbb B_f}(b) = 1$ otherwise. Restriction along $f$ clearly preserves ground states. To see that $f \colon \mathbb B_f \too \mathbb A$ is terminal among morphisms in $\Ising$
lying over $f \colon B \too A$, note that the (clearly unique) identity-on-$B$ from $\mathbb B' = (B,E_{\mathbb B'})$ to $\mathbb B_f = (B,E_{\mathbb B_f})$ is a well-defined morphism in $\Ising$ precisely because
the set of ground states of $\mathbb B_f$ as defined here is the \emph{minimal} one making $f$ a valid morphism in $\Ising$.
\end{proof}

A \emph{program} on an Ising net $\mathbb I$ in the sense of Definition~\ref{def-program} is a diagram of the form: 
$\bfig \morphism(0,0)/ >->/<500,0>[S`U\mathbb I;r] \morphism(1000,0)/ >->/<-500,0>[T`U\mathbb I;s] \efig$
By the previous lemma, such diagrams in $\Set$ give rise to cospans in $\Ising$ of the following form:

\noindent $\bfig \morphism(0,0)/->/<500,0>[\mathbb S_r`\mathbb I;r] \morphism(1000,0)/->/<-500,0>[\mathbb T_s`\mathbb I;s] \efig$. If $\C$ is a category with pushouts, we can form a bicategory $\Cospan(\C)$ whose objects are
cospans in $\C$ and in which composition is effected by pushout. We will show that although our category $\Ising$ does not have all pushouts, it still has enough for us to draw a link between cospan bicategories
and the composition of programs-on-Ising-nets given in Definition~\ref{def-composition}.

\begin{definition} Let us say that a span of the form: $ \bfig \morphism(1000,0)/->/<-500,0>[\mathbb T`\mathbb I;s] \morphism(1000,0)/->/<500,0>[\mathbb T`\mathbb J;t] \efig $ is \emph{admissible} if the intersection
of $s^{-1}(\gs \mathbb I)$ and $t^{-1}(\gs \mathbb J)$ is non-empty; that is, there must exist at least one ground state of $\mathbb T$ which is simultaneously the restriction along $s$ of a ground state of $\mathbb I$ and the
restriction of along $t$ of a ground state of $\mathbb J$.
\end{definition}

\begin{lemma}
%Consider a composable pair of cospans in $\Ising$, that is:
%\[ \bfig \morphism(0,0)/->/<500,0>[\mathbb S`\mathbb I;r] \morphism(1000,0)/->/<-500,0>[\mathbb T`\mathbb I;s] \morphism(1000,0)/->/<500,0>[\mathbb T`\mathbb J;t] \morphism(2000,0)/->/<-500,0>[\mathbb U`\mathbb J;u] \efig \]
%We should like to compose these in the usual way of cospans, that is, by forming the pushout of $s$ and $t$.
Let $\bfig \morphism(1000,0)/->/<-500,0>[\mathbb T`\mathbb I;s] \morphism(1000,0)/->/<500,0>[\mathbb T`\mathbb J;t] \efig$ be a span in $\Ising$. This span has a pushout in $\Ising$ if an only it is admissable, moreover,
this pushout can be obtained as:
$\bfig \square/->`->`->`->/[\mathbb T`\mathbb I`\mathbb J`\mathbb I \underset{T}{+} \mathbb J;s`t`\nu_I`\nu_J] \efig$
where $\nu_I \colon I \too I \underset{T}{+} J \oot J \colon \nu_J$ are the canonical pushout injections in $\Set$, and the energy function on $\mathbb I \underset{T}{+} \mathbb J$ is defined as
\[ E_{\mathbb I \underset{T}{+} \mathbb J}\left(s\right) = E_{\mathbb I}(s) + E_{\mathbb J}(s) \]
just as in Lemma~\ref{lem-gluing}.
\end{lemma}
\begin{proof} First we show that the given square is a pushout. The underlying maps of the square clearly commute; we must show that $\nu_I$ and $\nu_J$ are morphisms in $\Ising$ and we must verify the pushout property. 
Since the cospan $(s,t)$ is admissable, the nets $\mathbb I$ and $\mathbb J$ are compatible in the sense of Lemma~\ref{lem-gluing}. Thus, we can apply the results of that lemma, observing that the conclusions there
amount precisely to the assertions that $\nu_I$ and $\nu_J$ are morphisms in $\Ising$ and that the square is a pushout.
% Verifying the pushout property of the square in $\Ising$ follows easily from the pushout property of the underlying square in $\Set$
%and the observation that the set of ground states of $\mathbb I \underset{T}{+} \mathbb J$ is \emph{maximal} while still permitting $\nu_1$ and $\nu_2$ to be morphisms in $\Ising$.
Finally, suppose that cospan $(s,t)$ has a pushout in the form of a cospan $\bfig \morphism(1000,0)/<-/<-500,0>[\mathbb X`\mathbb I;u] \morphism(1000,0)/<-/<500,0>[\mathbb X`\mathbb J;v] \efig$ in $\Ising$. Then simply choose any ground
state $x \colon X \too \twoset$ and notice that $x \circ u \circ s = x \circ v \circ t$ is in the intersection of the $s$-restriction of the ground states of the ground states of $\mathbb I$ and the $t$-restriction of the
ground states of $\mathbb J$.
\end{proof}

\end{document}